\documentclass[12pt]{article}
\usepackage{amsmath, amsthm, amssymb, amstext}
\usepackage{mathrsfs}
\usepackage{array, amsfonts, mathrsfs}
\usepackage{latexsym, amsmath}
\usepackage{graphicx,color}
\usepackage{epsf}

\newtheorem{Lemma}{Lemma}
\newtheorem{Theorem}{Theorem}
\newtheorem{Convention}{Convention}

\newtheorem{Corollary}{Corollary}

\newtheorem{Definition}{Definition}

\title{Triangular factorization of operators and reconstruction of systems}
\author{M.I.Belishev\thanks {St.Petersburg Department of Steklov Mathematical
        Institute, St.Petersburg, Russia, e-mail: belishev@pdmi.ras.ru}.}

\date{}

\begin{document}

\maketitle

\begin{abstract}
The paper provides a coherent presentation of an operator scheme,
which is used in an approach to inverse problems of mathematical
physics (the boundary control method). The scheme is based on the
triangular factorization of operators. It not only solves inverse
problems but provides the functional models of a class of
symmetric semi-bounded operators. The class is characterized in
the terms of an evolutionary  dynamical system associated with the
operator.
\end{abstract}

\section{Introduction}\label{Sec Introd}

\subsubsection*{About the paper}

\noindent$\bullet$\,\,\, The purpose of the paper is a coherent
presentation of an operator scheme, which is used in an approach
to inverse problems of mathematical physics - the so-called
boundary control method (BC-method) \cite{B UMN,B IPI}. The scheme
is based on the triangular factorization (TF) of operators and, as
we hope, is of certain independent theoretical interest. It not
only solves inverse problems but provides the functional models of
a class of symmetric semi-bounded operators.

In the light of the system theory, to solve an inverse problem is
to recover a system via its input-output correspondence. However,
the only relevant understanding of `to recover' is to construct a
model system ({\it copy}) that responds to the inputs in the same
way as the original system \cite{KFA}. In particular, this means
to construct a model of the operator which governs the system.
Such a look clarifies the term {\it reconstruction} in the title.
\smallskip

\noindent$\bullet$\,\,\, The main subject of our approach is an
operator construction close to the classical triangular truncation
integral $\int PA\,dP$ and its version $\int dPA\,dP$
\cite{GK,Dav}. The construction appears and is used as follows.
Recall that a {\it nest} is a family of subspaces of a Hilbert
space, which is linearly ordered by the embedding $\subset$\,.
\smallskip

Let ${\mathscr F}$ and ${\mathscr H}$ be the Hilbert spaces,
$W:{\mathscr F}\to{\mathscr H}$ a bounded operator. Let $\mathfrak
f=\{{\mathscr F}_s\}_{s\in [0,T]}\subset{\mathscr F}$ be a
parametrized nest; by $X_s$ we denote the projections in
${\mathscr F}$ onto ${\mathscr F}_s$. The image $\mathfrak
h=W\mathfrak f:=\{{\mathscr H}_s\}_{s\in [0,T]}:\,{\mathscr
H}_s=\overline{W{\mathscr F}_s}$ is a nest in ${\mathscr H}$; let
$P_s$ project in ${\mathscr H}$ onto ${\mathscr H}_s$. The key
object is the operator $D_W:=\int_{[0,T]}dP_sWdX_s: {\mathscr
F}\to{\mathscr H}$, which is said to be a {\it diagonal} of $W$.
Compared to the classics, the difference is that two nests are
used: $\mathfrak f$ and $\mathfrak h$, with the latter being
determined by the former. Such a specifics is motivated by
applications, where the diagonal is the main player in solving the
inverse problems. For instance, it enables to reconstruct a
Riemannian manifold of arbitrary dimension and topology via
boundary inverse data \cite{BKach_1994,B Obzor IP 97,B IPI}.

\subsubsection*{Results}

\noindent$\bullet$\,\,\, Assume that $W:{\mathscr F}\to{\mathscr
H}$ is injective and possesses the diagonal $D_W$ provided
$\overline{{\rm Ran\,}D_{W}}=\overline{{\rm Ran\,}W}$. Let
$D_W^*=\Phi_{D_W^*}|D_W^*|$ be the polar decomposition of $D_W^*$.
We show (Lemma \ref{L V=D^*W triang}) that $W^*W=V^*V$ holds,
where the operator $V=\Phi_{D_W^*}W$ obeys $V{\mathscr
F}_s\subset{\mathscr F}_s,\,\,\,s\in[0,T]$, i.e., is {\it
triangular} with respect to the nest $\mathfrak f$.

As a consequence, taking ${\mathscr H}={\mathscr F}$ and
$W={\sqrt{C}}$, we arrive at the following result on TF (Theorem
\ref{T basic}). Let a positive injective operator $C$ be such that
the diagonal $D_{\sqrt{C}}$ does exist and provides
$\overline{{\rm Ran\,}D_{\sqrt{C}}}={\mathscr F}$. Then the
representation
\begin{equation}\label{Eq 0}
C=V^*V:\quad V\,=\,\Phi_{D^*_{\sqrt{C}}}\,{\sqrt{C}},\quad
V{\mathscr F}_s\subset{\mathscr F}_s,\,\,\, s\in[0,T]
\end{equation}
holds with the triangular factor $V$. Substantially, operator $C$
is not assumed to be boundedly invertible.

Thus, the diagonal recovers a triangular operator $V$ from its
module $|V|={\sqrt{C}}$, whereas the classical construction
determines a Volterra operator $A$ from its anti-Hermitian part by
$A=2i\,\int PA_J\,dP$ \cite{GK,Dav}.
\smallskip

\noindent$\bullet$\,\,\, In section \ref{Sec Models} we provide a
way to construct functional models of a class of symmetric
semi-bounded operators $L_0$ by means of TF. The class is
characterized by the properties of a dynamical system
$\alpha_{L_0}$, evolution of which is governed by $L_0^*$. To be
available for our approach, the system has to obey the finite
speed of propagation principle (FSPP) in a relevant form. System
$\alpha_{L_0}$ possesses an input-state correspondence $W$, which
acts from the outer space of inputs ${\mathscr F}$ to the inner
space of states ${\mathscr H}$. The key assumption is the
existence of the diagonal $D_W$ that is, roughly speaking, a form
of FSPP.

The system possesses the {\it connecting operator} $C=W^*W$, which
relates the Hilbert metrics in ${\mathscr F}$ and ${\mathscr H}$.
Its factorization (\ref{Eq 0}) provides a model (unitary copy) $V$
of the operator $W$, the model acting in the outer space
${\mathscr F}$. Along with $V$, the TF enables to construct in
${\mathscr F}$ a model $\tilde L_0$ of the basic operator $L_0$.
In applications, it is the model $\tilde L_0$ that solves inverse
problems. Its properties turn out to be very close to the ones of
the original $L_0$.

\subsubsection*{Comments}

\noindent$\bullet$\,\,\, The terminology we use is a mixture of
terms and definitions from \cite{GK} and \cite{Dav}. If otherwise
is not specified, all operators are bounded. Everywhere $\mathbb
O$ and $\mathbb I$ are the zero and unit operators. We say an
operator $C$ to be nonnegative or positive (and write $C\geqslant
\mathbb O$ or $C>\mathbb O$) if $(Cf,f)\geqslant 0$ or $(Cf,f)> 0$
holds for all $f\not=0$. The positive definiteness means
$C\geqslant\gamma\,\mathbb I$ with a $\gamma>0$. Everywhere
${\sqrt{C}}\geqslant\mathbb O$ is the square root of
$C\geqslant\mathbb O$.
\smallskip

\noindent$\bullet$\,\,\, We begin with an elementary case of the
matrix factorization, where the notion of the {\it orthogonalizer}
is introduced. Then in the finite dimensional case, we define an
operator diagonal and clarify its relation with the
orthogonalizer. It is a convenient way to prepare the further
infinite dimensional generalization.
\smallskip

\noindent$\bullet$\,\,\, The close relations between inverse
problems and triangular factorization are well known from the
middle of the XX century (M.G.Krein, I.M.Gelfand, B.M.Levitan,
V.A.Marchenko, L.D.Faddeev). The BC method develops these
relations and contributes to the understanding of the operator
theory background of them. Perhaps, the role of TF as a provider
of the operator models was first recognized in \cite{B DSBC IP
2001}.

\section{Finite dimension case}\label{Sec Fin Dim}

\subsubsection*{Matrix factorization}
We begin with an elementary well-known fact.
\smallskip

Let $e_1,\dots, e_n$ and $\langle\cdot,\cdot\rangle$ be the
standard basis and inner product in $\mathbb C^n$. By $O_n$ and
$I_n$ we denote the zero and unit $n\times n$ - matrices.
\begin{Lemma}\label{L Matrix TF}
Let $C=\{c_{ij}\}_{i,j=1}^n>O_n$ be a positive-definite matrix.
Then there is an (upper) triangular matrix
$V=\{v_{ij}\}_{i,j=1}^n:\,\, v_{ij}=0,\,\,i>j$ such that
\begin{equation}\label{Eq matrix TF}
C=V^*V
\end{equation}
holds. Assuming $v_{ii}>0,\,\,i=1,\dots,n$, such a matrix $V$ is
unique.
\end{Lemma}
\begin{proof}

\noindent{\bf 1.} \,\, Transforming $e_1,\dots, e_n\,\,
\mapsto\,h_1,\dots, h_n$ by the Schmidt procedure with respect to
the positive definite form $\langle Cf,g\rangle$, we construct a
basis
$$
h_1,\dots, h_n:\quad \langle Ch_i,h_j\rangle=\delta_{ij}, \qquad
i,j=1,\dots,n.
$$
The transformation is realized by the matrix  $\Psi: \Psi e_k=
h_k$, which is invertible and triangular: $\psi_{ij}=\langle \Psi
e_i,e_j\rangle=0,\,\,\,i>j$ holds. In the mean time, one has
$$
\delta_{ij}=\langle Ch_i,h_j\rangle=\langle C\Psi e_i,\Psi
e_j\rangle=\langle \Psi^*C\Psi e_i,e_j\rangle,\qquad i,j=1,\dots,n
$$
that implies $\Psi^*C\Psi =I_n$ and leads to
$C=(\Psi^{-1})^*\Psi^{-1}$. Thus, $V:=\Psi^{-1}$ provides (\ref{Eq
matrix TF}).
\smallskip

\noindent{\bf 2.}\,\,Let $C=W^*W$ be one more factorization of the
form (\ref{Eq matrix TF}). By $W^*W=V^*V$ we have $\langle
Wf,Wg\rangle=\langle Vf,Vg\rangle,\,\, f,g\in \mathbb C^n$, and
hence $U=WV^{-1}$ is a unitary matrix. Since $V$, $V^{-1}$ and $W$
are upper-triangular, matrix $U$ is also upper-triangular. Hence,
it has to be diagonal: $u_{kl}=e^{i\alpha_k}\delta_{kl},\,\,{\rm
Im\,}\alpha_k=0$ holds and describes the only nonuniqueness in
determination of the triangular factor. By imposing $v_{ii}>0$,
one eliminates the nonuniqueness and selects the unique factor
$V$.
\end{proof}
The matrix $\Psi$ used in the proof, is called an {\it
orthogonalizer}. In what follows, its operator and continual
counterparts are revealed.

\subsubsection*{Operator factorization}

Here an operator look at the matrix factorization, which is
suitable for generalization to infinite dimension case, is
provided.
\medskip

\noindent$\bullet$\,\,\, Let ${\mathscr F}$ be a finite
dimensional Hilbert space, ${\rm dim\,}{\mathscr F}=n\geqslant
2$;\,\, let $\mathfrak f=\{{\mathscr F}_0,\dots,{\mathscr F}_n\}$
and $\mathfrak p_\mathfrak f=\{X_0,\dots,X_n\}$:
\begin{align*}
\{0\}={\mathscr F}_0\subset{\mathscr F}_1\subset\dots\subset
{\mathscr F}_n={\mathscr F},\qquad \mathbb
O=X_0<X_1<\dots<X_n=\mathbb I
\end{align*}
be the nests of subspaces and the corresponding projections. A
basis $f_1,\dots,f_n$ obeying
\begin{equation*}
(f_i,f_j)=\delta_{ij};\quad {\rm
span\,}\{f_1,\dots,f_k\}={\mathscr F}_k,\quad k=1,\dots,n.
\end{equation*}
is said to be a {\it nest basis}. Since $f_k\in {\mathscr
F}_k\ominus{\mathscr F}_{k-1}$ and ${\rm dim\,}[{\mathscr
F}_k\ominus{\mathscr F}_{k-1}]=1$ holds, any another nest basis is
$e^{i\alpha_1}f_1,\dots,e^{i\alpha_n}f_n$ with real $\alpha_k$.
Also we put $\Delta X_k:=X_k-X_{k-1}$ and note that $\Delta
X_kf_k=f_k$ holds .
\begin{Lemma}\label{L Operator Fact}
Any operator $C>\mathbb O$ in ${\mathscr F}$ admits the
factorization along any nest $\mathfrak f$, i.e., can be
represented as
\begin{equation}\label{Eq operator TF}
C=V^*V;\,\,\quad V{\mathscr F}_k={\mathscr F}_k,\quad
k=0,1,\dots,n.
\end{equation}
\end{Lemma}
\begin{proof}
Endowing the operator $C$ with the matrix in a nest basis,
applying Lemma \ref{L Matrix TF}, and returning to operators, we
establish (\ref{Eq operator TF}).
\end{proof}

To select the canonical factor $V$ by Lemma \ref{L Matrix TF} one
imposes the condition $(Vf_i,f_i)>0,\,\,i=1,\dots,n$.
\smallskip

\noindent$\bullet$\,\,\, Preparing a generalization to the case
${\rm dim\,}{\mathscr F}=\infty$, we modify the matrix approach as
follows.
\smallskip

Let ${\mathscr H}$ be a Hilbert space, ${\rm dim\,}{\mathscr
H}\geqslant n$, and let $W:{\mathscr F}\to{\mathscr H}$ be an
injective operator. It determines the nest $\mathfrak
h=\{{\mathscr H}_1,\dots,{\mathscr H}_n\}\,=:\,W\mathfrak f$ in
${\mathscr H}$,
\begin{align*}
\{0\}={\mathscr H}_0\subset{\mathscr H}_1\subset\dots\subset
{\mathscr H}_n={\rm Ran\,} W\subset{\mathscr H},\quad {\mathscr
H}_k:=W{\mathscr F}_k
\end{align*}
and the corresponding projection nest $\mathfrak p_\mathfrak
h=\{P_1,\dots,P_n\}:\,\,\mathbb O=P_0<P_1<\dots<P_n$, where $P_k$
projects in ${\mathscr H}$ onto ${\mathscr H}_k$. Denote $\Delta
P_k:=P_k-P_{k-1}$. The operator
\begin{equation}\label{Eq def Diag}
D_W\,:=\,\sum\limits_{k=1}^n\Delta P_k\, W\,\Delta
X_k\,:\,\,{\mathscr F}\to{\mathscr H}
\end{equation}
is called a {\it diagonal} of $W$ (w.r.t. the nests $\mathfrak f$
and $\mathfrak h$). The definition easily implies
\begin{equation}\label{Eq PD=DX}
P_k D_W\,:=\,D_W X_k,\qquad k=0,1,\dots,n,
\end{equation}
i.e., the diagonal intertwines the nests $\mathfrak p_\mathfrak f$
and $\mathfrak p_\mathfrak h$.
\smallskip

Let $f_1,\dots,f_n$ be a nest basis of $\mathfrak f$. As is easy
to see, the elements
$$
h_k:=\nu_k\,\Delta P_k\, W\,\Delta X_kf_k;\quad \nu_k:=\|\Delta
P_k\, W\,\Delta X_kf_k\|_{\mathscr H}^{-1},\qquad k=1,\dots, n
$$
form a nest basis of $\mathfrak h$ in ${\rm Ran\,}W={\mathscr
H}_n$, whereas the constants $\nu_k$ do not depend on the choice
of $f_1,\dots,f_n$. A way to construct $h_1,\dots,h_n$ is to apply
the Schmidt orthogonalization to $Wf_1,\dots,Wf_n$ in ${\mathscr
H}$.

The operator $N:=\sum_{k=1}^n\nu_k\,\Delta P_k=N^*$ commutes with
$\mathfrak p_\mathfrak h$:
\begin{equation}\label{Eq NP=PN}
NP_k\,=\,P_kN,\qquad k=0,1,\dots,n,
\end{equation}
and is said to be a {\it weight operator}. It is positive in
${\mathscr H}$, whereas $N\upharpoonright{\mathscr H}_n$ is
positive definite in ${\mathscr H}_n$.

The operator ({\it orthogonalizer})
$$
\Psi :=ND_W\,\,=\sum_{j=1}^n\nu_j\,\Delta
P_j\,\,\sum\limits_{k=1}^n\Delta P_k\,W\,\Delta
X_k=\sum\limits_{k=1}^n\nu_k\,\Delta P_k\,W\,\Delta X_k
$$
relates the nest bases by $\Psi f_k=h_k$ and, as such, is an
isometry from ${\mathscr F}$ onto ${\mathscr H}_n={\rm Ran\,}
W\subset{\mathscr H}$. Therefore
\begin{equation}\label{Eq Isometry}
\Psi^*\Psi\,=\,\mathbb I_{\mathscr F}, \qquad \Psi\Psi^*\,=\,P_n.
\end{equation}
holds. By the latter equality and with regard to $N=N^*$, we have
\begin{equation*}
ND_WD_W^*N\,=\,N|D_W^*|^2N\,=\,P_n,
\end{equation*}
where
$$
D_W^*=\Phi_*|D_W^*|,\quad |D_W^*|:=\sqrt{D_WD_W^*}\,>\, \mathbb
O_{\mathscr H}
$$
is the polar decomposition of the adjoint diagonal. This leads to
the representations
\begin{equation}\label{Eq weight operator}
N\,=\,|D_W^*|^{-1},\qquad \Psi\,=\,|D_W^*|^{-1}D_W\,,
\end{equation}
where the inverse is understood as
$(|D_W^*|\upharpoonright{\mathscr H}_n)^{-1}$ on ${\mathscr H}_n$.

The relations
\begin{equation}\label{Eq Psi*=Phi D*}
\Psi^*=D_W^*N^*=D_W^*N=(\,\Phi_*|D_W^*|\,)\,|D_W^*|^{-1}=\Phi_*
\end{equation}
lead to the equality
\begin{equation}\label{Eq Psi*=Phi D*}
\Psi^*\,=\,\Phi_*:\,\,{\mathscr H}\to{\mathscr F},
\end{equation}
which clarifies the relationship between orthogonalizer and
diagonal. Later on we provide a relevant continual analog of this
relationship.
\smallskip

\noindent$\bullet$\,\,\, Multiplying in (\ref{Eq PD=DX}) by $N$,
with regard to (\ref{Eq NP=PN}) we have
\begin{equation*}%\label{Eq PD'=D'X}
NP_kD_W=P_kND_W=P_k\Psi\overset{(\ref{Eq PD=DX})}=ND_WX_k =\Psi
X_k,\quad k=0,1,\dots,n,
\end{equation*}
so that $P_k\Psi=\Psi X_k$ holds, i.e., the orthogonalizer
provides an isometry, which intertwines the nests.

As is easy to see, the relations
$$
\Psi{\mathscr F}_k={\mathscr H}_k,\quad \Psi^*{\mathscr
H}_k={\mathscr F}_k,\quad k=1,\dots, n
$$
are valid and hence $\Psi^*W\,{\mathscr F}_k={\mathscr F}_k$
holds. As a result, we get
\begin{equation}\label{Eq canon W*W}
(\Psi^*W)^*(\Psi^*W)\,=\,W^*\Psi\,\Psi^*W\overset{(\ref{Eq
Isometry})}=\,W^*P_n W=W^*W,
\end{equation}
and conclude that the operator
\begin{equation}\label{Eq canon W*W}
V\,:=\Psi^*W\,\overset{(\ref{Eq Psi*=Phi D*})}=\,\Phi_*W
\end{equation}
provides triangular factorization of the operator $W^*W$ in
${\mathscr F}$ along the nest $\mathfrak f$. Such a factorization
is referred to as canonical.
\smallskip

\noindent$\bullet$\,\,\, As a consequence, we arrive at the
following look at the matrix TF. Let
$D^*_{\sqrt{C}}=\Phi_*|D^*_{\sqrt{C}}|$ be the polar
decomposition.
\begin{Theorem}\label{T1}
Any operator $C>\mathbb O$ in a finite-dimensional space
${\mathscr F}$ admits a triangular factorization along any given
nest $\mathfrak f$. There is a canonical factorization of the form
\begin{equation}\label{Eq canon C>O}
C\,=\,V^*V,\quad V\,=\,\Phi_*{\sqrt{C}}.
\end{equation}
\end{Theorem}
\begin{proof}
Putting ${\mathscr H}={\mathscr F}$ and $W={\sqrt{C}}$ in the
previous considerations, we arrive at (\ref{Eq canon C>O}).
\end{proof}
The factorization in the form (\ref{Eq canon C>O}) is most
appropriate for meaningful generalizations.

\section{Infinite dimension case}

\subsubsection*{On continuity of nests}

\noindent$\bullet$\,\,\, Let ${\mathscr F}$ be a Hilbert space;
let $\mathfrak f=\{{\mathscr F}_s\}_{s\in[0,T]}$ and $\mathfrak
p_\mathfrak f=\{X_s\}_{s\in[0,T]}$ be a nest and the corresponding
projection nest in ${\mathscr F}$. The parametrization is such
that ${\mathscr F}_s\subset{\mathscr F}_{s'}$ and $X_s\leqslant
X_{s'}$ holds for $s<s'$. If ${\mathscr F}_0=\{0\}$ and ${\mathscr
F}_T={\mathscr F}$ is fulfilled, a nest is said to be bordered. A
nest is {\it continuous} if\,\, $\text{s\,-}\lim\limits_{r\to
s}X_r=X_s$ holds.
\begin{Convention}\label{Conv 1}
In what follows, the nest denoted by $\mathfrak f$ is assumed
continuous and bordered.
\end{Convention}

Let ${\mathscr H}$ be a Hilbert space and let $W:{\mathscr
F}\to{\mathscr H}$ be a bounded operator. The operator determines
the nest
$$
\mathfrak h=W\mathfrak f\,:=\,\{{\mathscr H}_s\}_{s\in[0,T]},\quad
{\mathscr H}_s:=\overline{W{\mathscr F}_s}
$$
and the corresponding projection nest $\mathfrak p_\mathfrak
h=\{P_s\}_{s\in[0,T]}$, in which $P_0=\mathbb O_{\mathscr H}$
holds and $P_T$ projects onto ${\mathscr H}_T=\overline{{\rm
Ran\,}W}$.
\smallskip

\noindent$\bullet$\,\,\, Recall Convention \ref{Conv 1} and
discuss the continuity properties of $\mathfrak h$. Also, recall
that any monotone family of projections does have a limit and
denote $P_{s}^\pm:=\lim\limits_{\sigma\to s\pm
0}P_\sigma$,\,\,\,\,$\Delta P^\pm _{s}:= \mp(P_{s}-P_{s}^\pm)$.
The break of $\mathfrak h$ at a point $s_0\in [0,T]$ means that at
least one of the differences $\Delta P^\pm _{s_0}$ is a nonzero
projection. If $\Delta P^-_{s_0}=0$\,\,\,($\Delta P^+_{s_0}=0$),
we say that the family is continuous from the left (right).
\begin{Lemma}\label{L Continuity}
The nest $\mathfrak h=W\mathfrak f$ is continuous from the left on
$(0,T]$.
\end{Lemma}
\begin{proof}
\noindent{\bf 1.\,\,}Assume that {${\mathscr F}_s\nsubseteq{\rm
Ker\,} W$ holds for all $s\in(0,T]$.}

{Let $s_0\in(0,T]$ and $\Delta P^-_{s_0}\not=\mathbb O_{\mathscr
H}$ hold.} Then, by $P_sP_{s'}=P_{\min\,(s,s')}$ we have
$$
(P_{s_0}-P_{s'})P_s=P_{s}-P_s=\mathbb O_{\mathscr
H},\quad(P_{s_0}-P_{s'})P_{s_0}=P_{s_0}-P_{s'}; \qquad 0<s<s'<s_0,
$$
which implies
\begin{equation}\label{Eq Delta P}
\Delta P^-_{s_0}\,P_s=
\begin{cases}
\mathbb O_{\mathscr H}, \qquad \,\,\, s<s_0;\cr \Delta P^-_{s_0},
\qquad s=s_0
\end{cases}.
\end{equation}

By virtue of $\Delta P^-_{s_0}\not=\mathbb O_{\mathscr H}$ and
${\rm Ran\,} \Delta P^-_{s_0}\subset{\mathscr
H}_{s_0}=\overline{W{\mathscr F}_{s_0}}$, there is an
$f\in{\mathscr F}_{s_0}$ such that $\Delta P^-_{s_0}Wf\not=0$
holds. By (\ref{Eq Delta P}), the function
$$
\phi(s):=\|\Delta P^-_{s_0}WX_s f\|, \qquad s\in[0,T]
$$
vanishes identically for $s<s_0$ and obeys
$$
\phi(s_0)=\|\Delta P^-_{s_0}WX_{s_0} f\|=\|\Delta
P^-_{s_0}P_{s_0}W f\|=\|\Delta P^-_{s_0}W f\|\not=0.
$$
Thus, it has a break at $s=s_0$. In the mean time, the continuity
of the nest $\{X_s\}$ obviously implies the continuity of $\phi$.
This contradiction yields $\Delta P^-_{s_0}=\mathbb O_{\mathscr
H}$, i.e., $P_{s_0}$ is continuous from the left at all
$s_0\in(0,T]$.
\medskip

\noindent{\bf 2.\,\,}Let the assumption { ${\mathscr
F}_{s_0}\nsubseteq{\rm Ker\,} W$ be cancelled}, so that there is a
positive $\sigma:=\sup\,\{s>0\,|\,\,{\mathscr F}_{s}\subseteq{\rm
Ker\,} W\}< T$. Then by ${\mathscr
F}_s\big|_{s<\sigma}\subset{\mathscr F}_{\sigma}$ we have
${\mathscr F}_{s}\subseteq{\rm Ker\,} W$ for all $s\leqslant
\sigma$ and, hence, ${\mathscr H}_s\big|_{s<\sigma}=\{0\}$, i.e.,
$P_s\big|_{s<\sigma}=\mathbb O_{\mathscr H}$. So it makes sense to
treat the continuity of $P_s$ for $s\in[\sigma,T]$ only. This case
is reduced to the previous one by replacing ${\mathscr F}_s$ with
${\mathscr F}_s\ominus{\mathscr F}_{\sigma}$.
\end{proof}

\noindent$\bullet$\,\,\, The following example shows that a
bounded $W$ can break a continuous nest, realizing the case
$\Delta P^+_{s_o}\not=\mathbb O$.

Let $\{X_s\}_{s\in[0,T]}$ be a continuous nest in ${\mathscr F}$.
Take the unit norm elements $\varphi\in{\mathscr F}$,
$\psi\in{\mathscr H}$ and put
$W:=(\,\cdot\,,\varphi)\,\psi:{\mathscr F}\to{\mathscr H}$. Recall
that $P_s$ projects in ${\mathscr H}$ onto ${\mathscr
H}_s=\overline{{\rm Ran\,} WX_s}$. Denote $
\sigma\,:=\,\sup\,\{s\geqslant 0\,|\,\,X_s\varphi=0\}$ and
consider the nontrivial case $0\leqslant\sigma<T$. It is easy to
see that
\begin{equation*}\label{Eq Ran WX}
{\rm Ran\,} WX_s\,=\,\begin{cases}
           \{0\}, & 0\leqslant s\leqslant\sigma;\cr
           {\rm span\,}\{\psi\}, & \sigma<s\leqslant T
              \end{cases}
\end{equation*}
holds and implies
\begin{equation*}\label{Eq Break Ps}
P_s\,=\,\begin{cases} \mathbb O_{\mathscr H}, & 0\leqslant
s\leqslant\sigma;\cr
           (\,\cdot\,,\psi)\,\psi, & \sigma<s\leqslant T
              \end{cases}\,\,\,\,,
\end{equation*}
so that we have $\Delta
P^+_{\sigma}=(\,\cdot\,,\psi)\,\psi\not=\mathbb O_{\mathscr H}$.
Note that the possible case $\sigma=0$ corresponds to the break of
$P_s$ at $s=0$.

\subsubsection*{Diagonal}

\noindent$\bullet$\,\,\, Here we define a continual analog of the
construction (\ref{Eq def Diag}). As above, $\mathfrak
h=W\mathfrak f=\{{\mathscr H}_s\}_{s\in[0,T]},\,\,{\mathscr
H}_s=\overline{W{\mathscr F}_s}$\,\, is the nest in ${\mathscr H}$
determined by a bounded operator $W:{\mathscr F}\to{\mathscr H}$
and $\mathfrak p_\mathfrak h=\{P_s\}_{s\in[0,T]}$ is the
corresponding projection nest.
\smallskip

Let $T<\infty$  and let
$$
\Xi=\{s_k\}_{k=0}^n:\quad 0=s_0<s_1<\dots<s_n=T
$$
be a partition of $[0,T]$ of the range
$r^{\Xi}:=\max\limits_{k=1,\dots,K}(s_k-s_{k-1})$. Denote $\Delta
X_{s_k}:=X_{s_k}-X_{s_{k-1}}$, $\Delta
P_{s_k}:=P_{s_k}-P_{s_{k-1}}$ and put
\begin{equation}\label{Eq Integral Sums}
D^{\Xi}_W\,:=\,\sum\limits_{k=1}^n\Delta P_{s_k}\,W\,\Delta
X_{s_k}\,.
\end{equation}
The operator $D_W: {\mathscr F}\to{\mathscr H}$,
\begin{equation*}
D_W\,:=\,\text{w\,-}\lim\limits_{r^\Xi\to 0}\,D^{\Xi}_W
=\int_{[0,T]}dP_s\,W\,dX_s
 \end{equation*}
is a relevant generalization of (\ref{Eq def Diag}); it is called
a {\it diagonal} of $W$. The limit is understood by Riemann: for
arbitrary $\varepsilon>0$ and elements $f\in{\mathscr
F},\,h\in{\mathscr H}$ there is $\delta>0$ such that
$$
|\,([D_W- D^{\Xi}_W]f,h)\,|\,<\,\varepsilon
$$
holds for any partition $\Xi$ of the range $r^{\Xi}<\delta$.

Note that this definition itself does not require the continuity
of $\mathfrak f$. Also note that such a limit does exist not for
every $W$ and (even continuous) nest $\mathfrak f$: see
\cite{Dav}. \,In what follows we deal with the affirmative case
and say that $W$ possesses diagonal.

Below we present some of the diagonal properties.
\smallskip

\noindent$\bullet$\,\,\, Begin with the following fact of general
character. In the formulation we deal with a pair of arbitrary
nests, not assuming $\mathfrak h=W\mathfrak f$.
\begin{Lemma}\label{L two nests}
Let $B:{\mathscr F}\to{\mathscr H}$ be a bounded operator and let
$\{X_s\}_{s\in[0,T]}$ and $\{P_s\}_{s\in[0,T]}$ be the projection
nests in ${\mathscr F}$ and ${\mathscr H}$ respectively. If the
integral $D[B]:=\int_{[0,T]}dP_sB\,dX_s$ converges then the
relation
\begin{equation}\label{Eq two nests}
\|D[B]\|\,\leqslant\,\|B\|
\end{equation}
holds.  If $B$ is compact and the nest $\{X_s\}_{s\in[0,T]}$ is
continuous then the integral does converge by the
 norm to $\mathbb O$.
\end{Lemma}
\begin{proof} {\bf 1.\,\,\,}For the integral sums (\ref{Eq Integral Sums}),
by the orthogonality $\Delta P_{s_k}{\mathscr H}\perp\Delta
P_{s_l}{\mathscr H}$ and $\Delta X_{s_k}{\mathscr F}\perp\Delta
X_{s_l}{\mathscr F}$ for $k\not= l$, we have
\begin{align}
\notag & \|\sum\limits_{k=1}^n\Delta P_{s_k}B\,\Delta
X_{s_k}\,f\|^2=\sum\limits_{k=1}^n\|\Delta P_{s_k}B\Delta
X_{s_k}\,f\|^2\leqslant\|B\|^2\sum\limits_{k=1}^n\|\Delta
X_{s_k}\,f\|^2\leqslant\\
\label{Eq estimate general} & \leqslant\|B\|^2\|f\|^2.
\end{align}
Passing to the limit, we get (\ref{Eq two nests}).
\smallskip

\noindent{\bf 2.\,\,\,}A partition $\Xi$ of the segment $[0,T]$
and elements $\varphi\in{\mathscr F}$, $\psi\in{\mathscr H}$
determine the operator $D^\Xi[(\cdot,\varphi)\,\psi]:=
\sum\limits_{k=1}^n\Delta P_{s_k}[(\cdot,\varphi)\,\psi]\Delta
X_{s_k}:\,{\mathscr F}\to{\mathscr H}$, for which we have
\begin{align*}
& D^\Xi[(\cdot,\varphi)\,\psi]f=\\
&=\sum\limits_{k=1}^n\Delta P_{s_k}[(\cdot,\varphi)\,\psi]\Delta
X_{s_k}f=\sum\limits_{k=1}^n(\Delta X_{s_k}f,\varphi)\,\Delta
P_{s_k}\psi=\sum\limits_{k=1}^n(f,\Delta X_{s_k}\varphi)\,\Delta
P_{s_k}\psi.
\end{align*}
Using the orthogonality arguments that has led to (\ref{Eq
estimate general}), we have
\begin{align*}
&\|D^\Xi[(\cdot,\varphi)\,\psi]f\|^2=\\
&=\sum\limits_{k=1}^n|(f,\Delta X_{s_k}\varphi)|^2\|\Delta
P_{s_k}\psi\|^2\leqslant \underset{k=1,\dots,n}\max |(f,\Delta
X_{s_k}\varphi)|^2\sum\limits_{k=1}^n\|\Delta P_{s_k}\psi\|^2\leqslant\\
&\leqslant \|f\|^2\underset{k=1,\dots,n}\max\|\Delta
X_{s_k}\varphi\|^2\|\psi\|^2= \|f\|^2
\delta^{\,\Xi}_\varphi\|\psi\|^2,
\end{align*}
where $ \delta^{\,\Xi}_\varphi:=\underset{k=1,\dots,n}\max
\|\Delta X_{s_k}\varphi\|^2$. By
$X_{s_k}\varphi=X_{s_{k-1}}\varphi\oplus\Delta X_{s_k}\varphi$ we
have
$$
\delta^{\,\Xi}_\varphi\,=\underset{k=1,\dots,n}\max(\|X_{s_k}\varphi\|^2-\|X_{s_{k-1}}\varphi\|^2).
$$
By continuity of $\{X_s\}$, the function $\|X_s\varphi\|^2$ is
uniformly continuous on $s\in[0,T]$, which implies
$\delta^{\,\Xi}_\varphi\to 0$ when $r^{\Xi}\to 0$.

Summarizing, we arrive at the estimate
\begin{equation}\label{Eq delta to 0}
\|D^\Xi[(\cdot,\varphi)\,\psi]\|\leqslant
\left[\delta^{\,\Xi}_\varphi\right]^{1\over
2}\|\psi\|\underset{r^{\Xi}\to 0}\rightarrow 0.
\end{equation}

\noindent{\bf 3.\,\,}Let
$B_p=\sum_{l=1}^p(\cdot,\varphi_l)\,\psi_l$ be a finite rank
operator; then
$$
D^{\Xi}[B_p]=\sum\limits_{k=1}^n\Delta P_{s_k}B_p\,\Delta
X_{s_k}=\sum\limits_{l=1}^pD^\Xi[(\cdot,\varphi_l)\,\psi_l]
$$
holds. By (\ref{Eq delta to 0}) we have
\begin{equation}\label{Eq Estimate 1}
\|D^{\Xi}[B_p]\|
\leqslant\sum\limits_{l=1}^p\delta^{\,\Xi}_{\varphi_l}\,\|\psi_l\|\underset{r^{\Xi}\to
0}\rightarrow 0.
\end{equation}

\noindent{\bf 4.\,\,} Denote
$D^{\Xi}[B]:=\sum\limits_{k=1}^n\Delta P_{s_k}B\,\Delta X_{s_k}$.
Fix an $\varepsilon>0$. Using the fact that $B$ is compact, let us
choose a finite rank approximation $B_p$ to provide
$\|B-B_p\|<{\varepsilon\over 2}$. Then choose a partition $\Xi$ of
the proper range $r^\Xi$ to provide
$\|D^{\Xi}[B_p]\|<{\varepsilon\over 2}$ (see (\ref{Eq Estimate
1})). As a result, we have
\begin{align*}
& \|D^{\Xi}[B]\|=\|D^{\Xi}[B]-D^{\Xi}[B_p]+D^{\Xi}[B_p]\|\leqslant
\|D^{\Xi}[B]-D^{\Xi}[B_p]\|+\|D^{\Xi}[B_p]\|\overset{(\ref{Eq estimate general})}\leqslant\\
&\leqslant\|B-B_p\|+\|D^{\Xi}[B_p]\|<{\varepsilon\over
2}+{\varepsilon\over 2}={\varepsilon}.
\end{align*}
Hence, the integral sums $D^{\Xi}[B]$ converge by the norm to
$D[B]=\mathbb O$.
\end{proof}
\begin{Corollary}\label{C0}
If exists, the diagonal of $W$ obeys $\|D_W\|\leqslant\|W\|$. If
$\mathfrak f$ is continuous and operator $W$ is compact then it
does possess the diagonal $D_W=\mathbb O$.
\end{Corollary}

The idea of the proof of Lemma \ref{L two nests} is taken from
\cite{GK}, Chap. 1, Lemma 5.1, where the equality
$\int_{[0,T]}dX_sB\,dX_s=\mathbb O$ (in our notation) is
established for continuous $\{X_s\}$ and compact $B$. Lemma \ref{L
two nests} shows that the close fact $D_B=\mathbb O$ remains true
in spite of possible discontinuity of $\{P_s\}$ noted in the
example below Lemma \ref{L Continuity}.
\smallskip

\noindent$\bullet$\,\,\, The evident relations $\Delta
X_{s_k}X_s=\mathbb O_{\mathscr F}$ and $P_s\Delta P_{s_k}=\mathbb
O_{\mathscr H}$ for $s<s_k$ easily imply
\begin{equation}\label{Eq PD=DX inf dim}
P_sD_W\,=D_W X_s\,=\,\int_{[0,s]}dP_t\,W\, dX_t\,, \qquad s\in
[0,T].
\end{equation}
Hence, as in the finite-dimensional case (\ref{Eq PD=DX}), the
diagonal intertwins  the nests.
\smallskip

\noindent$\bullet$\,\,\, The following fact is simple but
important.
\begin{Lemma}\label{L Intertwining}
Let $A$ and $B$ be the bounded self-adjoint operators in the
Hilbert spaces ${\mathscr F}$ and ${\mathscr H}$ respectively.
Assume that an operator $D:{\mathscr F}\to{\mathscr H}$ is bounded
and satisfies $\overline{{\rm Ran\,}D}={\mathscr H}$. Let
$D^*=\Phi_* |D^*|:{\mathscr H}\to{\mathscr F}$ be the polar
decomposition. Then the intertwining $DA=BD$ implies
\begin{equation}\label{Eq Intertwin}
A|D|=|D|A,\quad B|D^*|=|D^*|B, \quad  A\Phi_* =\Phi_*B.
\end{equation}
\end{Lemma}
\begin{proof}
\noindent{\bf 1.}\,\,\,By $A=A^*$ and $B=B^*$, the relation
\begin{equation}\label{Eq L1 1}
DA=BD
\end{equation}
implies $AD^*=D^*B$, whereas the latter relation leads to
\begin{equation}\label{Eq L1 2}
AD^*D=D^*BD
\end{equation}
In the mean time, multiplying (\ref{Eq L1 1}) by $D^*$ from the
left, we have
\begin{equation}\label{Eq L1 3}
D^*DA=D^*BD.
\end{equation}
Comparing (\ref{Eq L1 2}) with (\ref{Eq L1 3}), we get
$$
AD^*D=D^*DA,
$$
i.e., $A$ commutes with $D^*D=|D|^2$. Hence, $A$ commutes with
$|D|$ (see, e.g., \cite{BirSol}).
\smallskip

\noindent{\bf 2.}\,\,\, Quite analogously, the relation $DA=BD$
and its consequence $AD^*=D^*B$ lead to $BDD^*=DAD^*$ and
$DAD^*=DD^*B$. Thus, we have $PDD^*=DD^*B$, i.e.,
$B|D^*|^2=|D^*|^2B$. Hence $B$ and $|D^*|$ do commute.
\smallskip

\noindent{\bf 3.}\,\,\,The relation $AD^*=D^*B$, along with the
commutation $B|D^*|=|D^*|B$, leads to
$$
A\Phi_*|D^*|=\Phi_*|D^*|B=\Phi_*B|D^*|,
$$
so that $A\Phi_*|D^*|=\Phi_*B|D^*|$ holds. Therefore,
$A\Phi_*=\Phi_*B$ is valid on
$$
\overline{{\rm Ran\,}|D^*|}={\mathscr H}\ominus{\rm Ker\,}
|D^*|={\mathscr H}\ominus{\rm Ker\,} D^*=\overline{{\rm Ran\,}
D}={\mathscr H}
$$
by the density of ${\rm Ran\,} D$ in ${\mathscr H}$. Hence,
$A\Phi_*=\Phi_*B$ holds on ${\mathscr H}$.
\end{proof}
\begin{Corollary}\label{C1}
Let, in addition to the assumptions of Lemma \ref{L Intertwining},
the relation ${\rm Ker\,}D=\{0\}$ be valid. Then $\Phi$ and
$\Phi_*$ are the unitary operators, so that $A$ and $B$ turn out
to be unitarily equivalent.
\end{Corollary}
\begin{Corollary}\label{C2}
Let operator $W:{\mathscr F}\to{\mathscr H}$ possess the diagonal
$D_W=\Phi|D_W|$ satisfying ${\rm Ker\,}D_W=\{0\}$ and
$\overline{{\rm Ran\,}D_W}={\mathscr H}$. Then the nests
$\mathfrak f$ and $\mathfrak h=W\mathfrak f$ are unitarily
equivalent, the equivalence being realized by $P_s\Phi=\Phi X_s$
and $\Phi_*P_s=X_s\Phi_*$,\,\,$s\in[0,T]$.
\end{Corollary}
\noindent The latter follows from the intertwining (\ref{Eq PD=DX
inf dim}) and Lemma \ref{L Intertwining}.
\smallskip

In applications \cite{BKach_1994,B Obzor IP 97,B UMN}, the
following case is encountered and is of interest:  ${\rm
Ker\,}D_W\not=\{0\}$ and $\overline{{\rm Ran\,}D_W}={\mathscr H}$
holds, whereas $D_W^*$ is an isometry from ${\mathscr H}$ into
${\mathscr F}$. In such a case, we have
\begin{equation*}
D^*_WD_W\,=\,X', \qquad D_WD_W^*\,=\,\mathbb I_{\mathscr H},
\end{equation*}
where $X'$ projects in ${\mathscr F}$ onto ${\mathscr F}\ominus
{\rm Ker\,}D_W$.

\subsubsection*{Triangularity}

\noindent$\bullet$\,\,\, Recall that a bounded operator
$V:{\mathscr F}\to{\mathscr F}$ \,is {\it triangular} with respect
to the nest $\mathfrak f$ (we write $V\in\mathfrak T_\mathfrak f$)
if\, $V{\mathscr F}_s\subset{\mathscr F}_s$ holds. The latter is
equivalent to $ VX_s=X_s VX_s,\,\,\, s\in[0,T]$.

The class $\mathfrak T_\mathfrak f$ is a subalgebra of the bounded
operator algebra $\mathfrak B({\mathscr F})$ \cite{Dav}. Here it
is worth noting again that not every operator $V\in \mathfrak
T_\mathfrak f$ does possess the diagonal. In other words, the
transformation $V\mapsto D_V$ is well defined not on the whole
$\mathfrak T_\mathfrak f$ \cite{Dav}.
\begin{Lemma}\label{L V=D^*W triang} Let an operator $W:{\mathscr
F}\to{\mathscr H}$ possess the diagonal $D_W$ with respect to a
continuous nest $\mathfrak f=\{{\mathscr F}_s\}_{s\in[0,T]}$ and
$\overline{{\rm Ran\,} D_W}=\overline{{\rm Ran\,} W}$ holds. Let
$D_W^*=\Phi_*|D_W^*|$ be the polar decomposition of $D_W^*$. Then
the operator
\begin{equation}\label{Eq canon W*W continual}
V\,:=\,\Phi_*W:\,\,{\mathscr F}\to{\mathscr F}
\end{equation}
belongs to the class $\mathfrak T_\mathfrak f$ and provides
\begin{equation}\label{Eq W*W=V*V}
V^*V\,=\,W^*W.
\end{equation}
\end{Lemma}
\begin{proof}
The intertwining  (\ref{Eq PD=DX inf dim}) yields
$$
X_sD_W^*\,=\,D_W^*P_s,\qquad s\in [0,T]
$$
and by virtue of (\ref{Eq Intertwin}) leads to
$X_s\Phi_*=\Phi_*P_s$, which provides
\begin{align*}
V{\mathscr F}_s=\Phi_*W{\mathscr F}_s=\Phi_*P_sW{\mathscr
F}_s=X_s\Phi_*W{\mathscr F}_s\subset{\mathscr F}_s,\qquad s\in
[0,T].
\end{align*}
Hence, $V\in\mathfrak T_\mathfrak f$ holds. Next, we have
\begin{align*}
\overline{{\rm Ran\,} |D_W^*|}={\mathscr H}\ominus{\rm
Ker\,}|D_W^*|={\mathscr H}\ominus{\rm Ker\,} D_W^*=\overline{{\rm
Ran\,} D_W}=\overline{{\rm Ran\,} W},
\end{align*}
the latter equality being valid by assumption of the Lemma. Hence,
$\Phi_*$ maps $\overline{{\rm Ran\,} D_W}=\overline{{\rm Ran\,}
W}$ to ${\mathscr F}$ isometrically. Therefore,
$\Phi_*^*\Phi_*=P_T$ projects in ${\mathscr H}$ onto
$\overline{{\rm Ran\,} W}$. As a consequence, we get
\begin{align*}
V^*V=(\Phi_* W)^*\Phi_* W=W^*\Phi_*^*\Phi_*W=W^*P_TW=W^*W.
\end{align*}
\end{proof}

If $\overline{{\rm Ran\,} D_W}=\overline{{\rm Ran\,} W}$ does not
hold, the diagonal may not provide factorization. Indeed, if $W$
is compact then the equality $D_W=\mathbb O$ makes (\ref{Eq
W*W=V*V}) impossible.
\smallskip

Formula (\ref{Eq canon W*W continual}) provides a continual
version of (\ref{Eq canon W*W}) and relevant analog of the
orthogonalizer: as well as in (\ref{Eq Psi*=Phi D*}), the latter
is $\Psi=\Phi_*^*$. Definition (\ref{Eq weight operator}) gives
reason to regard $N=|D_{W^*}|^{-1}$ as a continual weight
operator.

\subsubsection*{Factorization}

\noindent$\bullet$\,\,\, Recall that an operator $C$ in ${\mathscr
F}$ admits the triangular factorization along a nest $\mathfrak f$
if the representation $C=V^*V$ holds with $V\in\mathfrak
T_\mathfrak f$. If exists, such a factorization is not unique: if
$I\in\mathfrak T_\mathfrak f$ is an isometry (i.e., $I^*I=\mathbb
I$ holds) then one has $(IV)^*IV=V^*I^*IV=V^*V=C$, so that
$V'=IV\in\mathfrak T_\mathfrak f$ is valid. Therefore, a question
on a canonical factorization arises. The answer is the following.
\begin{Theorem}\label{T basic}
Let $\mathfrak f$ be a continuous nest, $C>\mathbb O$, and let
${\sqrt{C}}$ possess the diagonal
\begin{equation*}
D_{\sqrt{C}}\,=\,\int_{[0,T]}d\tilde P_s\,{\sqrt{C}}\,dX_s,
\end{equation*}
where $\tilde P_s$ projects in ${\mathscr F}$ onto
$\overline{{\sqrt{C}}{\mathscr F}_s}$. Assume that $\overline{{\rm
Ran\,}D_{\sqrt{C}}}={\mathscr F}$ holds. Let
$D_{\sqrt{C}}^*=\Phi_*|D_{\sqrt{C}}^*|$ be the polar
decomposition.  Then $C$ admits the (canonical) factorization
\begin{align}\label{Eq Th basic 2}
C=V^*V,\qquad V\,=\,\Phi_*\,{\sqrt{C}}\,\in\mathfrak T_\mathfrak
f.
\end{align}
\end{Theorem}
\begin{proof}
By positivity of $C$ and ${\sqrt{C}}$ we have $\overline{{\rm
Ran\,}{\sqrt{C}}}={\mathscr F}$. Hence, $\overline{{\rm
Ran\,}{\sqrt{C}}}=\overline{{\rm Ran\,}D_{\sqrt{C}}}$ holds and we
can apply Lemma \ref{L V=D^*W triang} for $W={\sqrt{C}}$. As well
as in the Lemma, the condition $\overline{{\rm
Ran\,}D_{\sqrt{C}}}={\mathscr F}$ implies $\Phi_*^*\Phi_*=\mathbb
I_{\mathscr F}$. Applying Lemma \ref{L V=D^*W triang} and checking
$$
V^*V=[\Phi_*\,{\sqrt{C}}]^*\,[\Phi_*\,{\sqrt{C}}]={\sqrt{C}}\,[\Phi_*^*\,\Phi_*]\,{\sqrt{C}}={\sqrt{C}}\,{\sqrt{C}}=C,
$$
we arrive at the statement of the Theorem.
\end{proof}
\begin{Corollary}\label{C3}
Let $C=\mathbb I+B>0$ in ${\mathscr F}$, and let $B$ be a compact
operator. Then $C$ admits the canonical factorization of the form
(\ref{Eq Th basic 2}) along any continuous nest $\mathfrak f$.
\end{Corollary}
Indeed, in this case we have ${\sqrt{C}}=\mathbb I+B'$ with a
compact $B'$, whereas Lemma \ref{L two nests} provides
$D_{B'}=\mathbb O_{\mathscr F}$. As a result,
$$
D_{\sqrt{C}}=\int_{[0,T]}d\tilde P_s\,dX_s
$$
holds and $V=\Phi_*{\sqrt{C}}$ does provide the factorization.
However, the positivity of $\mathbb I+B$ is equivalent to its
positive definiteness. Therefore, the factorization problem for
such an operator is solved by the use of classical truncation
integral \cite{GK,Dav}. In this case, our novelty is just a
representation of the triangular factor in the form (\ref{Eq Th
basic 2}).
\smallskip

\noindent$\bullet$\,\,\, If operator $C$ is compact, we have
$D_{\sqrt{C}}=\mathbb O$ and the factorization (\ref{Eq Th basic
2}) fails. Nevertheless, in some cases it is possible to preserve
it by modifying the diagonal. In \cite{Bel_BC&WFC}, in course of
solving a concrete inverse problem, at heuristic level, the
construction of the form
$$
D'_{\sqrt{C}}=\int_{[0,T]}d\tilde P_s{\sqrt{C}}\,\partial\,dX_s
$$
is proposed with a proper operator $\partial$, which makes
${\sqrt{C}}\partial$ a bounded but not compact operator. Such a
correction provides
\begin{align*}
C=V^*V,\qquad V\,=\,\Phi'_*\,{\sqrt{C}}\,\in\mathfrak T_\mathfrak
f,
\end{align*}
where ${D'_{\sqrt{C}}}^*=\Phi'_*|{D'_{\sqrt{C}}}^*|$.

As illustration, this trick with $\partial=\frac{d}{dt}$ enables
to factorize
$$
C:=\int_0^1\min\{t,\tau\}\,(\,\cdot\,)(\tau)\,d\tau=V^*V, \quad
V=\int_t^1(\,\cdot\,)(\tau)\,d\tau
$$
in ${\mathscr F}=L_2[0,1]$ along the Volterra nest ${\mathscr
F}_s=\{f\in{\mathscr F}\,|\,\,{\rm supp\,}f\subset[0,s]\}$.

\section{Factorization and models of operators}\label{Sec Models}
Here we show that triangular factorization is a source of
functional models of the operators and dynamical systems.

\subsubsection*{Operators}

\noindent$\bullet$\,\,\,The class of the operators that we deal
with, is the following. We assume that $L_0$ is a closed densely
defined symmetric positive definite operator in a Hilbert space
${\mathscr H}$ with nonzero defect indexes; so that
$$
\overline{{\rm Dom\,} L_0}={\mathscr H};\quad
L_0\subset{L_0^*};\quad L_0\geqslant\gamma\,\mathbb
I,\,\,\,\,\gamma>0;\quad 1\leqslant
n_+^{L_0}=n_-^{L_0}\leqslant\infty
$$
holds. By virtue of $n_\pm^{L_0}\not=0$ such an operator is
necessarily unbounded. We denote ${\mathscr K}:={\rm Ker\,}
{L_0^*}$ and use the projection $P$ in ${\mathscr H}$ onto
${\mathscr K}$. Note that ${\rm dim\,}{\mathscr K} =n_\pm^{L_0}$
holds.

Let $L$ be the extension of $L_0$ by Friedrichs, so that
$$
L_0\subset L=L^*\subset L^*_0,\quad L\geqslant\gamma\,\mathbb I,
\quad {\rm Ran\,} L={\mathscr H}
$$
holds. Its inverse $L^{-1}$ is a bounded self-adjoint operator in
${\mathscr H}$.
\smallskip

\noindent$\bullet$\,\,\, The operators
$$
\Gamma_1:=L^{-1}{L_0^*}-\mathbb I,\quad\Gamma_2:=P{L_0^*};\qquad
{\rm Dom\,}\Gamma_{1,2}={\rm Dom\,}{L_0^*},\quad{\rm
Ran\,}\Gamma_{1,2}=\mathscr K
$$
are called {\it boundary operators}. The collection $\{{\mathscr
K};\Gamma_1,\Gamma_2\}$ constitutes the (canonical by M.Vishik)
{\it boundary triple} of the operator $L_0$. The relation (Green's
formula)
\begin{equation*}
({L_0^*} u,v)-(u,{L_0^*}
v)=(\Gamma_1u,\Gamma_2v)-(\Gamma_2u,\Gamma_1v),\qquad u,v\in{\rm
Dom\,}{L_0^*}
\end{equation*}
is valid (see, e.g., \cite{Vishik,BD_DSBC,MMM}).
\subsubsection*{Dynamics governed by $L_0$}

\noindent$\bullet$\,\,\,The boundary triple, in turn, determines a
{\it dynamical system with boundary control} (DSBC) of the form
\begin{align}
\label{Eq 1}& {\ddot u}+L_0^*u = 0  && {\rm in}\,\,\,{{\mathscr H}}, \,\,\,t\in(0,T);\\
\label{Eq 2}& u|_{t=0}={\dot u}|_{t=0}=0 && {\rm in}\,\,\,{{\mathscr H}};\\
\label{Eq 3}& \Gamma_1 u = f && {\rm in}\,\,\,{{\mathscr
K}},\,\,\,t\in[0,T],
\end{align}
where $\dot{(\,\,)}:=\frac{d}{dt}$; $f=f(t)$ is a ${\mathscr
K}$-valued function of time ({\it boundary control}). By
$u=u^f(t)$ we denote the solution ({\it state at the moment $t$}).

For controls $f$ of the class
$$
{\mathscr M}\,:=\,\{f\in C^\infty([0,T];{\mathscr K})\,|\,\,{\rm
supp\,}f\subset(0,T]\}
$$
of smooth controls vanishing near $t=0$, system (\ref{Eq
1})--(\ref{Eq 3}) does have a unique classical solution
$u^f(t)\in{\rm Dom\,}{L_0^*},\,\,\,t\geqslant 0$ \cite{BD_DSBC}.
Note that ${\mathscr M}$ is dense in $L_2([0,T];\mathscr K)$.
\smallskip

\noindent$\bullet$\,\,\,For short, we call (\ref{Eq 1})--(\ref{Eq
3}) just system $\alpha^T$; the ${\mathscr H}$ - valued function
$u^f(\,\cdot\,)$ is its {\it trajectory}. The following is the
system theory attributes and basic properties of $\alpha^T$.
\smallskip

\noindent{\bf $\ast$}\,\,\,The space of controls (inputs)
$\mathscr F:=L_2([0,T];\mathscr K)$ is an {\it outer space}. It
contains the smooth dense class $\mathscr M$ and the (bordered and
continuous) nest of delayed controls
$$
\mathfrak f=\{{\mathscr F}_s\}_{s\in[0,T]}:\quad\mathscr
F_s:=\{f\in{\mathscr F}\,|\,\,{\rm supp\,}f\subset[T-s,T]\},
$$
where $T-s$ is the delay and $s$ is the action time.
\smallskip

\noindent{\bf $\ast$}\,\,\,The space $\mathscr H$ is an {\it inner
space}, the states $u^f(t)$ are its elements. It contains the {\it
reachable sets} $\mathscr {\mathscr U}_s:=\{u^f(T)\in{\mathscr
H}\,|\,\,f\in{\mathscr F}_s\cap\mathscr M\},\,\,\,s\in[0,T]$, and
the nest of the subspaces
$$
\mathfrak h\,=\,\{{\mathscr H}_s\}_{s\in[0,T]}:\quad{\mathscr
H}_s:=\overline{{\mathscr U}_s}.
$$
\noindent{\bf $\ast$}\,\,\,In the system $\alpha^T$, the {\it
input-state map} is realized by a {\it control operator}
$W:\mathscr F\to {\mathscr H}$,
$$
W f:=u^f(T)
$$
well defined on ${\mathscr M}$. It may be unbounded, but is always
closable \cite{B DSBC_3}. By the latter, we assume operator $W$
closed on a proper ${\rm Dom\,}W\subset{\mathscr F}$. It relates
the outer and inner nests: $\mathfrak h=W\mathfrak f$ holds in the
sense that $\overline{W[{\mathscr F}_s\cap{\rm Dom\,}W]}=
{\mathscr H}_s,\,\,s\in[0,T]$ is valid.
\smallskip

\noindent{\bf $\ast$}\,\,\,Since $W$ is closed, the {\it
connecting operator} $C:={W}^*W\geqslant\mathbb O_{\mathscr F}$ is
well defined on a proper ${\rm Dom\,} C\subset{\mathscr F}$
\cite{BirSol}. It connects the outer and inner space metrics:
$$
(Cf,g)_{{\mathscr F}}=(Wf,Wg)_{\mathscr
H}=(u^f(T),u^g(T))_{\mathscr H}
$$
holds. The polar decomposition of the control operator is
$W=\Phi_W|W|=\Phi_W\sqrt{C}$.
\smallskip

\noindent{\bf $\ast$}\,\,\,The {\it input-output map} is
$R:{\mathscr F}\to{\mathscr F}$, $(Rf)(t):=\Gamma_2
u^f(t),\,\,t\in[0,T]$ which is originally defined on ${\mathscr
M}$ and closable..
\smallskip

\noindent{\bf $\ast$}\,\,\,Since the operator $L_0^*$ that governs
the evolution of $\alpha^T$, does not depend on time, the
relations
\begin{equation}\label{Eq steady state 1}
u^{\ddot f}(\cdot)\,=\,\ddot u^f(\cdot)\,\overset{(\ref{Eq
1})}=\,-\,L_0^*u^f(\cdot)
\end{equation}
hold, which is referred to as a {\it steady-state property} of the
system $\alpha^T$. The evident equality $\ddot {\mathscr
M}={\mathscr M}$, along with this property, follow to
\begin{align}
\notag & L_0^*{\mathscr U}_T={\mathscr U}_T; \qquad {\rm
graph\,}{L_0^{*\,T}}:=\{u^f(T),L_0^*u^f(T)\,|\,\,u^f(T)\in{\mathscr U}_T\}=\\
\label{Eq Graph}&\overset{(\ref{Eq steady state 1})}=
\{Wf,-\,W\ddot f\,|\,\,f\in{\mathscr M}\},
\end{align}
where ${L_0^{*\,T}}:=L_0^*\upharpoonright{\mathscr U}_T$ is the
part of $L_0^*$ in ${\mathscr H}_T$
\smallskip

\noindent$\bullet$\,\,\, Let us emphasize that all the above
introduced objects and properties are fully determined by the
operator $L_0$, so that the family of systems
$\{\alpha^T\}_{T\geqslant 0}$ is an attribute of $L_0$. Moreover,
this attribute is its unitary invariant in the following sense. If
$U:{\mathscr H}\to\tilde {\mathscr H}$ is unitary then the
operator $\tilde L_0=UL_0U^*$ determines the corresponding systems
$\tilde\alpha^T$ (with the same outer spaces ${\mathscr F}$)
evolving in $\tilde{\mathscr H}$, whose elements are related with
the ones of $\alpha^T$ via $U$. For instance, the trajectories are
related by $\tilde u^f(t)=Uu^f(t),\,\,\,t\geqslant 0$.

The aforesaid makes the notation $\alpha^T_{L_0}$ more reasonable.

\subsubsection*{Models}

\noindent$\bullet$\,\,\, The following definition describes a
class of symmetric semi-bounded operators that is available for
constructing models via triangular factorization. Note in advance
that it is more convenient to construct the models of $L_0^*$ but
not of $L_0$.
\begin{Definition}\label{Def 1}
Operator $L_0$ belongs to the class $\mathfrak L$ if for any $T>0$
the control operator $W$ of the system $\alpha^T_{L_0}$ is bounded
and injective, and its connecting operator $C$ admits the
factorization (\ref{Eq Th basic 2}) in ${\mathscr F}$ along the
nest $\mathfrak f=\{{\mathscr F}_s\}_{s\in[0,T]}$.
\end{Definition}
Let $L_0\in\mathfrak L$. For a fixed $T>0$, let $C=V^*V$ be the
factorization (\ref{Eq Th basic 2}). Introduce the operator
${\tilde L_0^{*\,T}}:{\mathscr F}\to{\mathscr F},\,\,{\rm
Dom\,}{\tilde L_0^{*\,T}}={\mathscr M}$ via its graph
\begin{align}
\label{Eq tilde Graph}{\rm graph\,}{\tilde
L_0^{*\,T}}\,:=\,\{Vf,-\,V\ddot f\,|\,\,f\in{\mathscr M}\}.
\end{align}
\begin{Lemma}\label{L last}
The operators ${L_0^{*\,T}}$ and ${\tilde L_0^{*\,T}}$ are
unitarily equivalent.
\end{Lemma}
\begin{proof} Let $W=\Phi_W|W|=\Phi_W{\sqrt{C}}$ and
$V=\Phi_V|V|=\Phi_V{\sqrt{C}}$ be the polar decompositions. Since
$C$ is injective, the injectivity holds for $W$ and $V$ and their
modules, whereas $\Phi_W:{\mathscr F}\to{\mathscr H}_T$ and
$\Phi_V:{\mathscr F}\to{\mathscr F}$ are the unitary operators.
So, $W=UV$ holds with a unitary $U=\Phi_W\Phi_V^*$. The latter,
along with (\ref{Eq Graph}), leads to ${\tilde
L_0^{*\,T}}=U^*{L_0^{*\,T}} U$.
\end{proof}

As a result, we get an operator ${\tilde L_0^{*\,T}}$, which is a
unitary copy ({\it model}) of the part ${L_0^{*\,T}}$ of the
operator $L_0^*$ acting in the inner space ${\mathscr H}$ (in its
subspace ${\mathscr H}_T$). The copy acts in the (functional)
outer space ${\mathscr F}=L_2([0,T];\mathscr K)$ and, thus,
provides a {\it functional model} of ${L_0^{*\,T}}$.
\smallskip

\noindent$\bullet$\,\,\, Increasing $T\to\infty$, we extend the
parts ${L_0^{*\,T}}$ in the inner space and construct more and
more informative models ${\tilde L_0^{*\,T}}$ in the outer space.
As a result, we obtain a unitary copy of the part $L_{0\,{\mathscr
U}}^*:=L_0^*\upharpoonright{\mathscr U}$ of the operator $L^*_0$
on ${\mathscr U}:={\rm span\,}\{{\mathscr
U}_T\,|\,\,T>0\}\subset{\mathscr H}$. Its copy $\tilde
L_{0\,{\mathscr U}}^*$ acts in ${\mathscr
F}:=L_2([0,\infty);\mathscr K)$ and is well defined on the class
$C^\infty_0((0,\infty);\mathscr K)$ of the smooth compactly
supported controls $f$ vanishing near $t=0$.

A separate question is to what extend does the model $\tilde
L_{0\,{\mathscr U}}^*$ represent the original operator $L_0^*$?
For a `good' answer, one has to impose more conditions on the
class $\mathfrak L$. In particular, to provide $\overline
{\mathscr U}={\mathscr H}$ (i.e., $L_{0\,{\mathscr U}}^*$ to be
densely defined in ${\mathscr H}$) it is necessary and sufficient
for $L_0$ to be a completely non-self-adjoint operator
\cite{BD_DSBC}. Also more restrictions may be required. As
example, the following result is established in \cite{BSim_2023}.
\begin{Theorem}\label{BSim}
Let $L_0\in\mathfrak L$ be a completely non-self-adjoint operator
in a Hilbert space ${\mathscr H}$ provided $n_{L_0}^\pm=1$. Assume
that, for any $T>0$, the connecting operator of the system
$\alpha^T_{L_0}$ is of the form $C=\mathbb I+K$, where $K$ is an
integral operator with a smooth kernel. Then $L_0$ is unitarily
equivalent to the minimal Schroedinger operator
$S_0=-\frac{d^2}{dt^2}+ q$ with a potential $q=q(t)$ in the space
$L_2[0,\infty)$.
\end{Theorem}
The minimality means that $S_0=\overline{S_0\upharpoonright
C^\infty_0(0,\infty)}$ holds. Omitting the detail, just note that
the model Schroedinger operator is realized as $S_0=({\tilde
L_{0\,{\mathscr U}}^*})^*$. Also note that the model potential is
recovered {\it locally}: the values of $q\,|_{\,0\leqslant
t\leqslant T}$ are determined by system $\alpha^T_{L_0}$ (its
connecting operator $C$), but no information about
$\alpha^{T'}_{L_0}$ with $T'>T$ is required.
\smallskip

\noindent$\bullet$\,\,\, The injectivity of $C$ in the Definition
\ref{Def 1} is not substantial because at every step of
constructing the model one can operate with the controls $f\in
{\mathscr F}\ominus{\rm Ker\,}C$.

In principle, the boundedness of $W$ can be also cancelled. It was
not assumed, when we introduced the systems $\alpha^T$. The reason
is that all constructions used in factorization (in particular, a
diagonal) can be generalized on unbounded case. It would require
more accurate handle with the operator domains of definition and
dealing with the generalized solutions to the problem (\ref{Eq
1})--(\ref{Eq 3}).

\subsubsection*{Comments}

As was noted in Introduction, the above proposed factorization
scheme reveals an operator background of the basic version of the
BC-method \cite{B UMN,B IPI}, which is an approach to the inverse
problems. Here are some comments on this approach.
\smallskip

\noindent$\bullet$\,\,\, There is a  bit of philosophy. In terms
of the system theory, inverse problems are the problems of
determination of a system (e.g., $\alpha^T_{L_0}$) from its
`input$\to$output' correspondence. This correspondence is given by
some {\it inverse data} (e.g., operator $R$), which formalize the
measurements implemented by the external observer. Given data, the
observer must recover the system: its structure, parameters, etc
(in particular, recover operator $L_0$).

However, the original system is placed in the inner space
${\mathscr H}$, which is invisible and unreachable for the
observer. In such a situation, by the very general system theory
thesis, the only relevant understanding of `to recover' is to
create an isomorphic copy of the system extracting it from the
inverse data (see \cite{KFA}, section 10.6). The copy has to be
constructed from the objects of the outer space of inputs
(controls) ${\mathscr F}$, in which the observer operates. System
$\tilde\alpha^T_{L_0}$ is such a copy.
\smallskip

\noindent$\bullet$\,\,\, A cornerstone of the BC-method is that a
wide class of the inverse data, in the time (as our $R$) or
frequency domains, does determine the connecting operator $C$ of
the system under investigation. Moreover, $C$ is expressed via the
inverse data in a simple and explicit form \cite{B UMN,B IPI}. Due
to that, the observer is able to realize the factorization
(\ref{Eq Th basic 2}) and get the factor $V$ which {\it
visualizes} the invisible states $u^f$ by $\tilde u^f=Vf=\Phi u^f$
with an isometric $\Phi$. Also, by the use of (\ref{Eq tilde
Graph}) the observer constructs the copy $\tilde L_0$ of the
original $L_0$. In accordance with the above mentioned thesis,
this is the maximum which the observer can hope for.

A remarkable fact is that, in applications, the copies $\tilde
u^f$ turn out to be very informative and close to the originals
$u^f$ that makes the term {\it visualization} quite motivated
\cite{B IPI,BSim_2023}. It is the fact, due to which the
triangular factorization is a tool for solving inverse problems.
\smallskip

\noindent$\bullet$\,\,\, Factorization is most relevant for the
`hyperbolic' systems $\alpha^T_{L_0}$, in which the states (waves)
propagate with finite velocity. This holds for important
applications in acoustics, geophysics, electrodynamics,
elasticity. Recently the factorization scheme was realized in
numerical experiment: in \cite{Pestov Film} a movie is
demonstrated, in which a scattering process into an inhomogeneous
media is visualized in the real time.
\smallskip

\noindent$\bullet$\,\,\, At heuristic level, the key construction
of the diagonal was proposed in \cite{Bel_BC&WFC} and then
rigorously justified in \cite{BKach_1994,B Obzor IP 97}.

\bigskip

\noindent{\bf Key words:}\,\,\,triangular factorization of
operators, nest theory, functional models, inverse problems.
\smallskip

\noindent{\bf MSC:}\,\,\,47Axx,\,\, 47B25,\,\, 35R30.

\end{document}